\PassOptionsToPackage{dvipsnames}{xcolor}
\ifdefined\XeTeXversion\PassOptionsToClass{nopdfoutputerror}{quantumarticle}\fi
\documentclass[a4paper,onecolumn,11pt,accepted=2026-11-03]{quantumarticle}
\ifdefined\pdfoutput\pdfoutput=1\fi
\ifdefined\XeTeXversion\else
  \usepackage[utf8]{inputenc}
\fi
\usepackage[T1]{fontenc}
\usepackage[english]{babel}
\usepackage{amsmath,amsfonts,amssymb,amsthm}
\usepackage{thm-restate}
\usepackage[ruled,vlined,linesnumbered,commentsnumbered]{algorithm2e}
\SetKwInput{KwIn}{Input}
\SetKwInput{KwOut}{Output}
\SetKwInput{KwGlobal}{Global variable}
\SetKwInput{KwOracle}{Oracle}
\SetKwComment{Comment}{// }{}
\usepackage{graphicx}
\usepackage{enumerate}
\usepackage{bbm}
\usepackage{verbatim}
\usepackage{hyperref,color}
\usepackage[capitalize,nameinlink]{cleveref}
\usepackage{xcolor}
\usepackage[numbers]{natbib}
\hypersetup{
	colorlinks=true,
	pdfpagemode=UseNone,
	citecolor=OliveGreen,
	linkcolor=NavyBlue,
	urlcolor=Magenta,
	pdfstartview=FitW
}
\usepackage{appendix}
\usepackage{makecell}
\usepackage{tablefootnote}
\crefname{appsec}{Appendix}{Appendices}
\usepackage{tikz}
\usetikzlibrary{positioning}
\usepackage{pgfplots}
\usepackage{xifthen}
\usepackage{subcaption}
\usepackage{hhline}

\makeatletter
\ifdefined\newcounteralias
  \renewcommand{\thmt@autorefsetup}{%
    \expandafter\def\csname\thmt@envname autorefname\expandafter\endcsname
      \expandafter{\thmt@thmname}}
\fi
\makeatother

\theoremstyle{plain}
\newtheorem{theorem}{Theorem}[section]
\newtheorem{proposition}[theorem]{Proposition}
\newtheorem{lemma}[theorem]{Lemma}

\theoremstyle{definition}

\newtheorem*{assumption*}{Assumption}

\theoremstyle{remark}
\newtheorem{remark}[theorem]{Remark}

\crefname{lemma}{Lemma}{Lemmas}
\crefname{theorem}{Theorem}{Theorems}
\crefname{definition}{Definition}{Definitions}
\crefname{fact}{Fact}{Facts}
\crefname{claim}{Claim}{Claims}
\crefname{proposition}{Proposition}{Propositions}

\DeclareMathOperator{\Haf}{Haf}
\DeclareMathOperator{\Perm}{Perm}

\newcommand{\dist}{\mathrm{dist}}

\newcommand{\eps}{\varepsilon}

\newcommand{\defeq}{:=}

\renewcommand{\epsilon}{\varepsilon}

\newcommand{\abs}[1]{\left\vert#1\right\vert}

\def\*#1{\boldsymbol{#1}} % Use \*A for \mathbf{A}
\def\+#1{\mathcal{#1}} % Use \+A for \mathcal{A}
\def\-#1{\mathrm{#1}} % Use \-A for \mathrm{A}
\def\=#1{\mathbb{#1}} % Use \^A for \mathbb{A}
\def\!#1{\mathfrak{#1}} % Use \!A for \mathfrak{A}

\newcommand\boxprod{\mathbin{\text{\scalebox{.84}{$\square$}}}}

\def\oPr{\mathop{\mathrm{Pr}}}
\renewcommand{\Pr}[2][]{ \ifthenelse{\isempty{#1}}
  {\oPr\left[#2\right]}
  {\oPr_{#1}\left[#2\right]} }

\def\oE{\mathop{\mathbb{E}}}
\newcommand{\E}[2][]{ \ifthenelse{\isempty{#1}}
  {\oE\left[#2\right]}
  {\oE_{#1}\left[#2\right]} }

\def\oVar{\mathrm{Var}}
\newcommand{\Var}[2][]{ \ifthenelse{\isempty{#1}}
  {\oVar\left[#2\right]}
  {\oVar_{#1}\left[#2\right]} }

\def\oEnt{\mathrm{Ent}}
\newcommand{\Ent}[2][]{ \ifthenelse{\isempty{#1}}
  {\oEnt\left[#2\right]}
  {\oEnt_{#1}\left[#2\right]} }

\newcommand{\PhiEnt}[2][]{ \ifthenelse{\isempty{#1}}
  {\oEnt^\phi\left[#2\right]}
  {\oEnt^\phi_{#1}\left[#2\right]} }

\title{Simulating Gaussian boson sampling on graphs in polynomial time}
\author[1]{Konrad Anand}
\author[2]{Zongchen Chen}
\author[3]{Mary Cryan}
\author[3]{Graham Freifeld}
\author[1]{Leslie Ann Goldberg}
\author[3]{Heng Guo}
\author[2]{Xinyuan Zhang}
\affil[1]{Department of Computer Science, University of Oxford, Wolfson Bldg, Parks Rd, Oxford OX1 3QD, United Kingdom}
\affil[2]{School of Computer Science, Georgia Institute of Technology, North Avenue, Atlanta, Georgia 30332, USA}
\affil[3]{School of Informatics, University of Edinburgh, Informatics Forum, Edinburgh, EH8 9AB, United Kingdom}

\pgfplotsset{compat=1.18} 
\begin{document}

\begin{abstract}
  We show that a distribution related to Gaussian Boson Sampling (GBS) on graphs can be sampled classically in polynomial time. 
  Graphical applications of GBS typically sample from this distribution, and thus quantum algorithms do not provide exponential speedup for these applications.
  We also show that another distribution related to Boson sampling can be sampled classically in polynomial time.
\end{abstract}

\maketitle

\section{Introduction}

Bosons are subatomic particles with quantum spins. Boson sampling (BS) 
is a model of quantum computation based on interacting bosons that can be implemented using an optical network. It
was proposed by Aaronson and Arkhipov \cite{AA13} as a model that may have quantum advantage in the sense that sampling approximately from the output distribution can be done on a Boson computer. However, assuming two conjectures - the permanent anti-concentration conjecture (concerning the permanent of a matrix of i.i.d.\ Gaussians) - and the permanent of Gaussian conjecture (that estimating this permanent is 
\#P-hard), they show that classically sampling from the same distribution would collapse the polynomial hierarchy.
Later, Gaussian boson sampling (GBS), a variant of BS which avoids some of the costs of producing photon sources, was proposed by Hamilton, Kruse, Sansoni, Barkhofen, Silberhorn, and Jex \cite{HKSBSJ17}.
GBS has gained popularity quickly, and various teams have made considerable efforts to build photonic quantum computers based on it \cite{GBS-Science,GBS-PRL,GBS-Xanadu22,GBS-Xanadu25}.

The main proposed applications of GBS are to solve hard graph problems, such as counting perfect matchings \cite{BDRSW18}, finding densest subgraphs \cite{AB18},  graph isomorphisms \cite{BFIKS21}, planted bipartite cliques \cite{CMT25}, and graph colourings \cite{EBM26}. 
The implementation in \cite{GBS-PRL} 
uses samples from GBS for related graph search problems. 

We start by defining the notation that we will use to sample from the GBS distribution on graphs (Theorem~\ref{thm:main-GBS}). While our presentation is self-contained, we mostly use the notation of~\cite{ZZWWYYXL25}. We will subsequently extend to a related Boson sampling setting with non-negative weights (Theorem~\ref{thm:main-BS}).
Given a graph~$G=(V,E)$ and a subset $S 
\subseteq V$, let $|S|$ denote the size of~$S$ and let
$PM(S)$ denote the number of perfect matchings in the induced subgraph~$G[S]$. There is a positive real number~$c$
 --- specifically $c$ is the inverse of the maximum norm of an eigenvalue of the adjacency matrix of~$G$. Then the GBS distribution corresponding to~$G$ is defined as follows:
\begin{align} \label{eqn:GBS-G-intro} 
\mu_{GBS,G}(S)\propto c^{2\abs{S}}PM(S)^2.
\end{align}
In GBS applications the number $c$ is in the range $(0,1)$ but in this paper we will not make any assumptions except that $c$ is real and positive.

While there are certainly fast quantum algorithms for sampling from the GBS distribution  --- this is the whole point of GBS --- it was an open question whether there is a fast classical algorithm, even in the unweighted (graph) case.  
Zhang, Zhou, Wang, Wang, Yang, Yang, Xue, and Li \cite{ZZWWYYXL25} showed 
recently that there is a fast algorithm for very dense graphs.  In particular, they \cite[Theorem 5]{ZZWWYYXL25} showed that
for any real number $\xi\geq 1$ 
there is an algorithm which samples from $\mu_{GBS,G}$ in time $\widetilde{O}(n^{2\xi+18})$ for $n$-vertex graphs with minimum degree at least $n-\xi$. Note that this is a very severe restriction on the density of the graph.

A related result which is interesting for graph applications is the quantum-inspired algorithm by Oh, Fefferman, Jiang, and Quesada \cite{OFJQ24}. 
Motivated by the application of GBS to graph problems and by the desire to develop efficient classical algorithms, they proposed sampling from a distribution that is a little different from $\mu_{GBS,G}$, in the sense that 
the probability of the output~$S$ is proportional to $PM(S)$ rather than to $PM(S)^2$ as in \eqref{eqn:GBS-G-intro}. They found that, in practice, their classical algorithm (from the different distribution) did not perform much worse in experiments than GBS.
Their distribution can be sampled by two-photon boson sampling, which can be simulated classically (or can be directly sampled classically using the algorithm of Jerrum and Sinclair; see Proposition~\ref{prop:JS89}). Nevertheless, their work left open the question of whether there is an efficient classical algorithm for GBS in the non-negative case.

Our main result is that there is indeed a classical polynomial-time algorithm to sample from~$\mu_{GBS,G}$ (for all graphs).

\begin{restatable}{theorem}{mainGBS}\label{thm:main-GBS} 
  There is an algorithm that, given a graph $G=(V,E)$ with associated real~$c>0$, and arbitrary~$\epsilon$, samples from a distribution that is $\epsilon$-close to $\mu_{GBS,G}$ in total variation distance, in time $O(\overline{c}mn^4 \log^2 (n\overline{c}/\epsilon))$ where $m=\abs{E}$, $n=\abs{V}$, and $\overline{c}=\max\{1,c\}$.
\end{restatable}

Note that Theorem~\ref{thm:main-GBS} works for any $c>0$, not just for $c\in(0,1)$ as in the GBS application.

\begin{remark}
As a consequence of Theorem \ref{thm:main-GBS}, we have shown that there is no exponential-time quantum speedup for any application based on sampling from the distribution $\mu_{GBS,G}$.
\end{remark}

Our  algorithm to sample from $\mu_{GBS,G}$  is different from the approach in \cite{ZZWWYYXL25}. In fact, our algorithm is simpler, and faster, and it   also applies to all graphs. 
We first construct the Cartesian product
$G\boxprod K_2$ of $G$ with an edge.
We then note that the distribution~$\mu_{GBS,G}$
can be obtained by first sampling a perfect matching from $G\boxprod K_2$ with an appropriate weight function on edges
and then projecting this perfect matching onto the vertex set matched by the original edges of $G$.
To sample weighted perfect matchings from $G\boxprod K_2$, we use the  Markov chain of Jerrum and Sinclair~\cite{JS89}.
The key property which guarantees that the 
Jerrum--Sinclair (JS) chain generates a perfect matching in polynomial time is that the ratio between the number of near-perfect matchings 
of $G
\boxprod K_2$
(matchings with two unmatched vertices) and the number of perfect matchings of $G\boxprod K_2$ is bounded above by a polynomial in the size of the graph.
We show that, for any graph~$G$, this property holds for $G\boxprod K_2$.

In addition to Gaussian boson sampling, we consider the original boson sampling model~\cite{AA13}. Here (using the notation of~\cite{CC18}),  the input is an $m\times n$ matrix $A$ with $n \leq m$, where $A$ is the first $n$ columns of a random $m\times m$ unitary matrix.
Let $\Phi_{m,n}$ denote the set of non-negative integer vectors $\mathbf{z}=\{z_1,...,z_m\}$ such that $\sum_{i=1}^m z_i=n$.
The boson sampling distribution is
\begin{align}\label{eqn:BS-intro}
    \forall \mathbf{z}\in\Phi_{m,n},\quad\quad\mu_{BS}(\mathbf{z})=\frac{\abs{\Perm(A_{\mathbf{z}})}^2}{\prod_{i=1}^m z_i!},
\end{align}
where $\Perm(\cdot)$ is the permanent function and $A_{\mathbf{z}}$ is the
square matrix composed of $z_i$~copies of row~$i$ of~$A$. 
Exponential time classical simulation algorithms \cite{CC18,CC24} are known for this distribution, 
but it is conjectured that no polynomial-time classical approximate sampler exists \cite{AA13}.
Our second main result shows that, if the input is a non-negative matrix   instead of a truncated unitary matrix, then one can efficiently sample from a similar distribution,
where we modify the definition in \eqref{eqn:BS-intro} by 
renormalising appropriately, but retaining the fact that
$\forall \mathbf{z}\in\Phi_{m,n}$, $\mu_{BS}(\mathbf{z})\propto\frac{\abs{\Perm(A_{\mathbf{z}})}^2}{\prod_{i=1}^m z_i!}$.

\begin{restatable}{theorem}{mainBS}\label{thm:main-BS} 
There is an algorithm that, given a non-negative real $m\times n$ matrix $A$ with $n<m$ and a real number $\epsilon \in (0,1)$,  
samples from a distribution that is $\eps$-close to $\mu_{BS}$ in total variation distance, in time $O\left(\frac{m^7n^{14}}{\eps^7}\log^4\left(\frac{m n}{\eps}\right)\right)$.
\end{restatable}

Our technique for proving Theorem~\ref{thm:main-BS} is a combinatorial reduction similar to the proof of Theorem \ref{thm:main-GBS}.
Given the matrix~$A$,
we construct a bipartite graph $G$ such that 
a certain distribution on weighted perfect matchings of~$G$  induces a distribution that is close to $\mu_{BS}$.
We then use the algorithm of Jerrum, Sinclair, and Vigoda~\cite{JSV04} to efficiently sample from the distribution on perfect matchings.  

When $A$ is not a non-negative matrix, the Jerrum--Sinclair--Vigoda algorithm (JSV) is no longer applicable.
However our construction would still work, if $\abs{\Perm(A_{\mathbf{z}})}$ could be efficiently approximated for each $\mathbf{z}$.
The task of approximating the norm of these permanents is thus the main source of difficulty for classically simulating boson sampling. The problem of approximating the permanent 
is \textbf{NP}-hard in general, even for real positive semi-definite matrices \cite{Mei23}.

In Section \ref{sec:prelim} we 
define the terminology and
review some known efficient classical algorithms for sampling matchings in graphs. 
We first prove Theorem \ref{thm:main-BS} in Section \ref{sec:BS} as a warmup,
and then prove Theorem \ref{thm:main-GBS} in Section \ref{sec:GBS}.

\section{Preliminaries}\label{sec:prelim}

\subsection{Classical algorithm to sample (perfect) matchings} 

Let $G=(V,E)$ be a graph. 
A \emph{matching} $M$ of~$G$ is a subset of edges such that no two edges of $M$ share an endpoint.
Let $V_M$ denote the vertex set of $M$.
If $V_M=V$, then $M$ is a \emph{perfect matching}. If $|V_M| = |V|-2$ then $M$ is a \emph{near-perfect matching}.
Denote by $\+M$ the set of all matchings of~$G$, and by $\+M_k$ the set of matchings of size $k$.
Then, when $\abs{V}$ is even, the set of perfect matchings is $\+M_{\abs{V}/2}$.
Let $(\lambda_e)_{e\in E}$ be a collection of (non-negative) weights on the edges in~$E$.
Define the following distribution over the matchings of~$G$:
\begin{align}\label{eqn:matching}
    \forall M\in\+M,\quad \quad\mu_{matching,\lambda}(M)\propto \prod_{e\in M}\lambda_e.
\end{align}

Jerrum and Sinclair \cite{JS89} showed  
that there is a polynomial-time algorithm to approximately sample from the distribution \eqref{eqn:matching}. To explain this more precisely, we need the following notation.
The total variation (TV) distance between two distributions~$\mu$ and~$\nu$ over some discrete space~$\Omega$ is defined as 
\begin{align*}
    \dist_{TV}(\mu,\nu)\defeq\frac{1}{2}\sum_{\omega\in\Omega}\abs{\mu(\omega)-\nu(\omega)}.
\end{align*}

\begin{proposition}\label{prop:JS89} (Jerrum and Sinclair \cite{JS89})
There is an algorithm that, given $G=(V,E)$, weights $(\lambda_e)_{e\in E}$, and 
a real number $\epsilon \in (0,1)$,
samples from a distribution that is $\eps$-close to $\mu_{matching,\lambda}$ in TV distance, in time $O(\overline{\lambda}m n^2 \log \frac{n \overline{\lambda}}{\eps})$, 
where $m=\abs{E}$, $n=\abs{V}$, and $\overline{\lambda}=\max_{e\in E}\{1,\lambda_e\}$.
\end{proposition}
This particular running time is derived from \cite[Proposition 12.4]{JS96}.

Proposition~\ref{prop:log-concave} shows that
the number of matchings of size $k$ form a log-concave sequence (see \cite[Theorem 5.1]{JS89}).

\begin{proposition}\label{prop:log-concave} 
    For any graph $G=(V,E)$ and $1\le k\le \abs{V}/2$, $\abs{M_{k-1}}\abs{M_{k+1}}\le \abs{M_k}^2$.
\end{proposition}

Later, to prove \eqref{eqn:PM-percent}, we will need a weighted version of Proposition~\ref{prop:log-concave}, due to Heilmann and Lieb \cite{HL72}.
See Proposition~\ref{prop:Z-log-concave}.

Note that Proposition~\ref{prop:JS89} samples weighted matchings but not \emph{perfect} matchings.
To efficiently sample perfect matchings, one can use the Jerrum-Sinclair algorithm only if $\frac{\abs{M_{t-1}}}{\abs{M_{t}}}$ is bounded by a polynomial of the input size, where $t=\abs{V}/2$.
When this is the case, due to log-concavity (Proposition~\ref{prop:log-concave}),
one can tune the edge weights so that perfect matchings show up sufficiently frequently.
This will be our strategy for Theorem~\ref{thm:main-GBS} later.

On the other hand, in bipartite graphs, Jerrum, Sinclair, and Vigoda \cite{JSV04} showed that
one can efficiently sample from the distribution \eqref{eqn:matching} restricting to perfect matchings, without further assumptions.
The running time of their algorithm was subsequently improved by Bez{\'{a}}kov{\'{a}}, \v{S}tefankovi\v{c}, Vazirani, and Vigoda \cite{BSVV08}.
More precisely, for any graph $G=(V,E)$ and any $M \in \+M_{\abs{V}/2}$, let 
\begin{align}\label{eqn:PM-dist}
     \mu_{PM,\lambda}(M)\propto \prod_{e\in M}\lambda_e.
\end{align}
\begin{proposition}\label{prop:JSV04}
(Jerrum, Sinclair, Vigoda \cite{JSV04};
Bez{\'{a}}kov{\'{a}}, \v{S}tefankovi\v{c}, Vazirani, and Vigoda \cite{BSVV08})
There is an algorithm that, given a bipartite $G=(V,E)$, weights $(\lambda_e)_{e\in E}$, and a real number $\epsilon \in (0,1)$, 
samples from a distribution that is $\eps$-close to $\mu_{PM,\lambda}$ in TV distance, in time $O(n^7\log^4n+n^5\log\frac{n}{\eps})$, 
where $n=\abs{V}$.
\end{proposition}
This running time comes from \cite{BSVV08}. 
One needs to first spend $O(n^7\log^4n)$ time to estimate certain weights, and then each subsequent sampling takes $O(n^5\log\frac{n}{\eps})$ time. 
See \cite[Theorem 1.1 and Theorem 4.1]{BSVV08}.

\subsection{Permanent and Hafnian}

Two matrix functions will be useful throughout this paper --- the permanent and the hafnian.
Let $A$ be an $n\times n$ matrix.
The permanent is defined as follows
\begin{align*}
    \Perm(A)\defeq \sum_{\sigma\in\+{S}_{n}} \prod_{i=1}^n A_{i,\sigma(i)},
\end{align*}
where $S_{n}$ is the symmetric group of order $n$.
If $A$ is the biadjacency matrix of some bipartite graph $G$, 
then $\Perm(A)$ is the number of perfect matchings in $G$.

Similarly, for a $2n\times 2n$ symmetric matrix $A$, 
\begin{align*}
    \Haf(A) \defeq \frac1{2^n n!} \sum_{\sigma \in \mathcal S_{2n}} \prod_{i=1}^n A_{\sigma(2i-1),\sigma(2i)}.
\end{align*}
If $A$ is the adjacency matrix of some graph $G$ with $2n$ vertices, 
then $\Haf(A)$ is the number of perfect matchings of $G$.

\section{Boson sampling}\label{sec:BS}

Boson sampling was proposed by Aaronson and Arkhipov \cite{AA13} as a way to demonstrate quantum advantage.
For exact classical simulation, see \cite{CC18,CC24}.
Mathematically, the distribution to sample from is the following.
Let $A$ be an $m\times n$ matrix with $n\le m$. 
Let $\mathbf{z}=\{z_1,\ldots,z_m\}$ be a vector of non-negative integers with $\sum_{i=1}^{m} z_i=n$,
and let $\Phi_{m,n}$ denote the set of these vectors.
Then,
\begin{align}\label{eqn:BS}
    \forall \mathbf{z}\in \Phi_{m,n},\quad\quad\mu_{BS}(\mathbf{z})\propto \frac{\abs{\Perm(A_\mathbf{z})}^2}{\prod_{i=1}^m z_i !},
\end{align}
where $A_{\mathbf{z}}$ is the square matrix composed of $z_i$ copies of row $i$ of $A$.
Physically, the matrix $A$ is usually the first $n$ columns of an $m\times m$ Haar random unitary matrix, in which case the normalising factor in \eqref{eqn:BS} is $1$.

In the rest of this section, we prove Theorem~\ref{thm:main-BS}, which we re-state for convenience.

\mainBS*

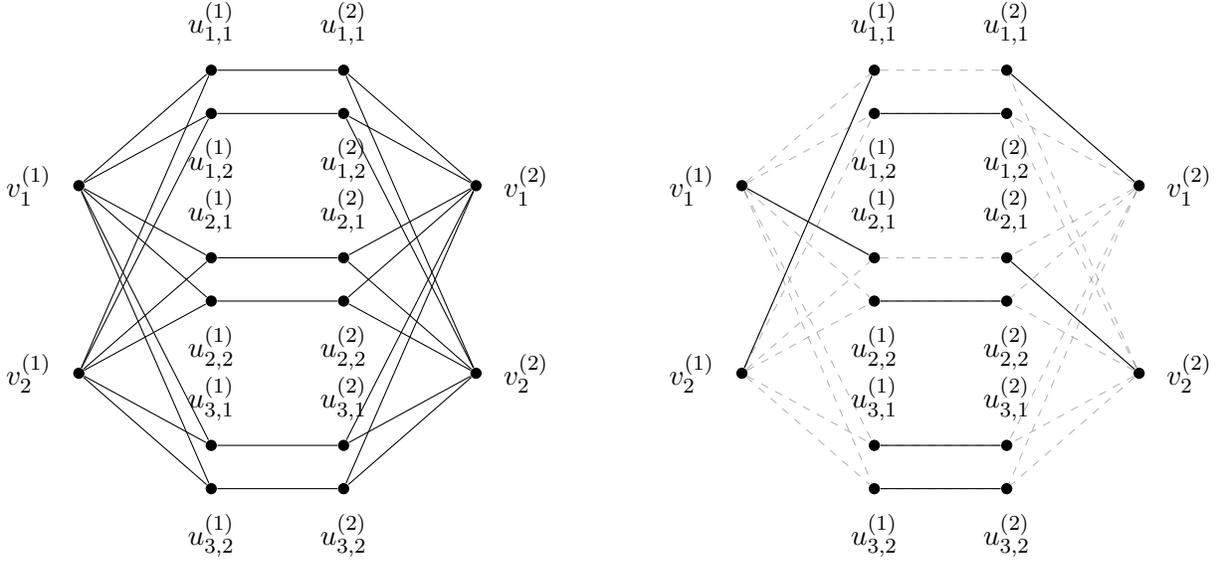
\begin{figure}[htbp]
    \centering
\begin{tikzpicture}[
    every node/.style={circle, fill=black, inner sep=1.5pt},
    xscale=2.2*0.8, yscale=1.2*0.8]

% ==== Column 1: v_i ====
\node (v1) at (0,1.3) {};
\node (v2) at (0,-1.3) {};

\node[draw=none, left=3pt of v1, fill=none] {$v_1^{(1)}$};
\node[draw=none, left=3pt of v2, fill=none] {$v_2^{(1)}$};

% ==== Column 2: u_jk ====
\node (u11) at (1,2.9) {};
\node (u12) at (1,2.3) {};
\node (u21) at (1,0.3) {};
\node (u22) at (1,-0.3) {};
\node (u31) at (1,-2.3) {};
\node (u32) at (1,-2.9) {};

\node[draw=none, above=1pt of u11, fill=none] {$u_{1,1}^{(1)}
$};
\node[draw=none, below=1pt of u12, fill=none] {$u_{1,2}^{(1)} $};
\node[draw=none, above=1pt of u21, fill=none] {$u_{2,1}^{(1)} $};
\node[draw=none, below=1pt of u22, fill=none] {$u_{2,2}^{(1)} $};
\node[draw=none, above=1pt of u31, fill=none] {$u_{3,1}^{(1)} $};
\node[draw=none, below=1pt of u32, fill=none] {$u_{3,2}^{(1)} $};

% Edges from v_i to u_jk
\foreach \i in {1,2} {
    \foreach \j in {1,2,3} {
        \foreach \k in {1,2} {
            \draw (v\i) -- (u\j\k);
        }
    }
}

% ==== Column 2: u_jk' ====

\node (u11p) at (2,2.9) {};
\node (u12p) at (2,2.3) {};
\node (u21p) at (2,0.3) {};
\node (u22p) at (2,-0.3) {};
\node (u31p) at (2,-2.3) {};
\node (u32p) at (2,-2.9) {};

\node[draw=none, above=1pt of u11p, fill=none] {$u_{1,1}^{(2)}$};
\node[draw=none, below=1pt of u12p, fill=none] {$u_{1,2}^{(2)}$};
\node[draw=none, above=1pt of u21p, fill=none] {$u_{2,1}^{(2)}$};
\node[draw=none, below=1pt of u22p, fill=none] {$u_{2,2}^{(2)}$};
\node[draw=none, above=1pt of u31p, fill=none] {$u_{3,1}^{(2)}$};
\node[draw=none, below=1pt of u32p, fill=none] {$u_{3,2}^{(2)}$};

% Edges between second and third columns
\foreach \j in {1,2,3} {
    \foreach \k in {1,2} {
        \draw (u\j\k) -- (u\j\k p);
    }
}

% ==== Column 4: u' and v' ====
\node (v1p) at (3,1.3) {};
\node (v2p) at (3,-1.3) {};

\node[draw=none, right=3pt of v1p, fill=none] {$v_1^{(2)}$};
\node[draw=none, right=3pt of v2p, fill=none] {$v_2^{(2)}$};

% Edges from u' to all u_i', and from v' to all v_i'
% Edges from v_i to u_jk
\foreach \i in {1,2} {
    \foreach \j in {1,2,3} {
        \foreach \k in {1,2} {
            \draw (v\i p) -- (u\j\k p);
        }
    }
}

\end{tikzpicture}
\hspace{1cm}
\begin{tikzpicture}[
    every node/.style={circle, fill=black, inner sep=1.5pt},
    xscale=2.2*0.8, yscale=1.2*0.8]

% ==== Column 1: v_i ====
\node (v1) at (0,1.3) {};
\node (v2) at (0,-1.3) {};

\node[draw=none, left=3pt of v1, fill=none] {$v_1^{(1)}$};
\node[draw=none, left=3pt of v2, fill=none] {$v_2^{(1)}$};

% ==== Column 2: u_jk ====
\node (u11) at (1,2.9) {};
\node (u12) at (1,2.3) {};
\node (u21) at (1,0.3) {};
\node (u22) at (1,-0.3) {};
\node (u31) at (1,-2.3) {};
\node (u32) at (1,-2.9) {};

\node[draw=none, above=1pt of u11, fill=none] {$u_{1,1}^{(1)}
$};
\node[draw=none, below=1pt of u12, fill=none] {$u_{1,2}^{(1)} $};
\node[draw=none, above=1pt of u21, fill=none] {$u_{2,1}^{(1)} $};
\node[draw=none, below=1pt of u22, fill=none] {$u_{2,2}^{(1)} $};
\node[draw=none, above=1pt of u31, fill=none] {$u_{3,1}^{(1)} $};
\node[draw=none, below=1pt of u32, fill=none] {$u_{3,2}^{(1)} $};

% Edges from v_i to u_jk
\foreach \i in {1,2} {
    \foreach \j in {1,2,3} {
        \foreach \k in {1,2} {
            \draw[draw opacity=0.3, dashed] (v\i) -- (u\j\k);
        }
    }
}

% ==== Column 2: u_jk' ====

\node (u11p) at (2,2.9) {};
\node (u12p) at (2,2.3) {};
\node (u21p) at (2,0.3) {};
\node (u22p) at (2,-0.3) {};
\node (u31p) at (2,-2.3) {};
\node (u32p) at (2,-2.9) {};

\node[draw=none, above=1pt of u11p, fill=none] {$u_{1,1}^{(2)}$};
\node[draw=none, below=1pt of u12p, fill=none] {$u_{1,2}^{(2)}$};
\node[draw=none, above=1pt of u21p, fill=none] {$u_{2,1}^{(2)}$};
\node[draw=none, below=1pt of u22p, fill=none] {$u_{2,2}^{(2)}$};
\node[draw=none, above=1pt of u31p, fill=none] {$u_{3,1}^{(2)}$};
\node[draw=none, below=1pt of u32p, fill=none] {$u_{3,2}^{(2)}$};

% Edges between second and third columns
\foreach \j in {1,2,3} {
    \foreach \k in {1,2} {
        \draw[draw opacity=0.3, dashed] (u\j\k) -- (u\j\k p);
    }
}

% ==== Column 4: u' and v' ====
\node (v1p) at (3,1.3) {};
\node (v2p) at (3,-1.3) {};

\node[draw=none, right=3pt of v1p, fill=none] {$v_1^{(2)}$};
\node[draw=none, right=3pt of v2p, fill=none] {$v_2^{(2)}$};

% Edges from u' to all u_i', and from v' to all v_i'
% Edges from v_i to u_jk
\foreach \i in {1,2} {
    \foreach \j in {1,2,3} {
        \foreach \k in {1,2} {
            \draw[draw opacity=0.3, dashed] (v\i p) -- (u\j\k p);
        }
    }
}
\draw (v1) -- (u21);
\draw (v2p) -- (u21p);
\draw (v2) -- (u11);
\draw (v1p) -- (u11p);
\draw (u12) -- (u12p);
\draw (u22) -- (u22p);
\draw (u31) -- (u31p);
\draw (u32) -- (u32p);
\end{tikzpicture}
\caption{On the left we have our graph $G$ with vertex sets from left to right: $L_1, R_1, R_2, L_2$. The corresponding matrix $A$ is a $3\times 2$ matrix with all $1$ entries, and $k=2$ in this example. On the right, a perfect matching $M$ selected from $G$ is highlighted. Here $S_1(M) = \{u_{1,1}^{(1)},u_{2,1}^{(1)}\}$ and $\mathbf{z}=(1,1,0)$.}
\label{fig:G-construct}
\end{figure}

\begin{proof}
Given the matrix~$A$, we will construct a bipartite graph $G$ and weights $\lambda$ such that its distribution $\mu_{PM,\lambda}$ in \eqref{eqn:PM-dist} is $\eps/2$-close to 
the distribution $\mu_{BS}$ in \eqref{eqn:BS}. We will
then finish by invoking Proposition~\ref{prop:JSV04} with error $\eps/2$.

Let 
$k = \lceil 4 n^2/\epsilon \rceil$.
We define $G$ as follows:
for each $\ell \in \{1,2\}$  
let $L_\ell = \{ v_j^{(\ell)} : j\in [n]\}$.
Let $R_\ell = \{ u_{i,t}^{(\ell)}: i \in [m], t\in [k] \}$.
The vertex set of $G$ is $L_1 
\cup R_1 
\cup L_2 
\cup R_2$.
The edge set of $G$ is defined as follows:
For each $\ell\in \{1,2\}$ let $E_\ell = \{(v_j^{(\ell)}, u_{i,t}^{(\ell)}): i\in [m], j\in [n], t\in [k], A_{i,j} \neq 0\}$.
For each $e = (v_j^{(\ell)}, u_{i,t}^{(\ell)})\in E_\ell$, $\lambda_e = A_{i,j}$.
Let 
$E_R = \{(u_{i,t}^{(1)},u_{i,t}^{(2)}): j
\in [m], t\in [k]\}$.
For each $e\in E_R$, $\lambda_e = 1$.
The edge set of $G$ is $E_1 \cup E_2 \cup E_R$.
An example is given in Figure \ref{fig:G-construct}.

Let $\mu_{PM,\lambda}$ from 
\eqref{eqn:PM-dist} be the distribution over perfect matchings in $G$ with weights given by $\lambda$. 
Let $M\sim\mu_{PM,\lambda}$ be a sample from this distribution. 
Let $S_1(M)$ be the set of vertices in $R_1$ 
that are matched to vertices in $L_1$ by~$M$ (rather than to vertices in $R_2$). 
    Note that, because every vertex in $L_1$ needs to be matched with some vertex in $R_1$, $\abs{S_1(M)}=\abs{L_1}=n$.
    We identify $S_1(M)$ with a vector $\mathbf{z}\in\Phi_{m,n}$,
    by 
defining $z_i$ as follows for each $i\in [m]$: 
$z_i =  \abs{S_1(M) \cap\left\{u_{i,1}^{(1)},\ldots,u_{i,k}^{(1)}\right\}}$.
    Let $\nu$ be the resulting distribution over $\Phi_{m,n}$.
    We claim that, for $k=\lceil \frac{4n^2}{\eps}\rceil$,
    \begin{align}\label{eqn:closeness}
        \dist_{TV}(\nu,\mu_{BS}) \le \frac{\eps}{2}.
    \end{align}
    The theorem follows from the claim and Proposition~\ref{prop:JSV04}.

    In the rest of the proof we show \eqref{eqn:closeness}.
    For any particular $\mathbf{z}\in\Phi_{m,n}$,
    the way to obtain the sample $\mathbf{z}$ is to first choose $z_i$ 
    copies of each $u_i$ from 
    $\{u_{i,1}^{(1)},\ldots,u_{i,k}^{(1)}\}$ - denote the set of these by $S_1$ -
    and then match $S_1$ with $L_1$ perfectly using edges in $E_1$.
    Moreover, the vertices in $R_1 \setminus S_1$ must be matched through $E_R$ edges (which have weight~$1$) to $R_2$.
    As a result, the vertices in $R_2$ with the same indices as those in $S_1$ must be matched via edges in $E_2$ to $L_2$. Let $S_2$ denote the set of these vertices in $R_2$.
    The total weight of matching $L_1$ with $S_1$ perfectly is  $\Perm(A_\mathbf{z})$
    since the rows of $A_\mathbf{z}$ contain $z_i$ copies of row~$i$ of~$A$. The matching from $L_2$ to $S_2$ is independent of this and also has total weight
    $\Perm(A_\mathbf{z})$.
    To summarise,
    \begin{align}
        \nu(\mathbf{z})& \propto k^{-n} \nu(\mathbf{z})\propto k^{-n}\prod_{i=1}^m\binom{k}{z_i}\Perm(A_{\mathbf{z}})^2=\frac{\Perm(A_\mathbf{z})^2}{\prod_{i=1}^m z_i !}\cdot k^{-n}\prod_{i=1}^m\frac{k!}{(k-z_i)!}\notag\\
        &= \frac{\Perm(A_\mathbf{z})^2}{\prod_{i=1}^m z_i !}\cdot k^{-n}\prod_{i=1}^m \prod_{j=0}^{z_i-1} (k-j) = \frac{\Perm(A_\mathbf{z})^2}{\prod_{i=1}^m z_i !}\cdot \prod_{i=1}^m\prod_{j=0}^{z_i-1}\left(1-\frac{j}{k}\right). \label{eqn:nu-dist}
\end{align} 
Furthermore, using the facts that $k=\lceil \frac{4n^2}{\eps}\rceil $ and $z_i \leq k$ and $\sum_{i=1}^m z_i=n$,
    \begin{align*}
        1\ge \prod_{i=1}^m\prod_{j=0}^{z_i-1}\left(1-\frac{j}{k}\right) \ge \left(1-\frac{n}{k}\right)^n \ge e^{-2n^2/k} = e^{-\eps/2}.
    \end{align*}
    Plugging the estimate above into \eqref{eqn:nu-dist}, we have that the multiplicative error for each $\mathbf{z}$ between the proportional weights of $\nu$ and $\mu_{BS}$ is between $1$ and $e^{-\eps/2}$,
    which implies \eqref{eqn:closeness}.
\end{proof}

\section{Gaussian boson sampling} \label{sec:GBS}
Gaussian boson sampling (GBS) is a variant proposed by Hamilton, Kruse, Sansoni, Barkhofen, Silberhorn, and Jex \cite{HKSBSJ17}. See also \cite{KHSBSJ19}.
For $n$ photons and $m$ modes, the input is a Gaussian state characterised by a $2m\times 2m$ covariance matrix $\sigma$. 
Let $A$ be the \emph{sampling matrix} given by $A=
\begin{pmatrix}
    0 & I_m\\
    I_m & 0
\end{pmatrix}\left(I_{2m}-\sigma_Q^{-1}\right)$,
where $I_m$ denotes the $m\times m$ identity matrix and $\sigma_Q\defeq\sigma + {I_{2m}}/{2}$.
For any $\mathbf{z}\in \Phi_{m,n}$, the probability of the output pattern $\mathbf{z}$ is
\begin{align}\label{eqn:GBS}
    \mu_{GBS}(\mathbf{z})\defeq\frac{\Haf (A_{\mathbf{z}})}{\sqrt{\det (\sigma_Q)}\prod_{i=1}^m z_i!},
\end{align}
where $A_{\mathbf{z}}$ is the square matrix corresponding to $\mathbf{z}$.
We will consider the case where $n=O(\sqrt{m})$, where, with high probability, $z_i\le 1$ (see \cite[Section B]{KHSBSJ19}), and in this case the matrix $A_{\mathbf{z}}$ is the principal submatrix of~$A$ formed by keeping 
rows and columns~$i$ and $m+i$ for each $i\in [m]$ with $z_i=1$.\footnote{In the general case, the forming of $A_z$ is more complicated, and this is not needed for our paper. Details can be found in \cite[Section III]{KHSBSJ19}.}

Br\'adler, Dallaire-Demers, Rebentrost, Su, and Weedbrook \cite{BDRSW18} proposed using GBS to solve graph problems.
They showed that one can choose the covariance matrix $\sigma$ and $c>0$ such that, for any graph $G$ with adjacency matrix $A$, the distribution in \eqref{eqn:GBS} can be interpreted as a distribution over $S\subseteq V$:
\begin{align}\label{eqn:GBS-G}
    \mu_{GBS,G}(S) \propto c^{2|S|} \Haf(A_S)^2,
\end{align}
where $A_S$ is the principal submatrix of the adjacency matrix $A$ corresponding to $S$.
Note that $\Haf(A_S)$ counts the number of perfect matchings in the induced subgraph $G[S]$.   In fact, the distribution induced by GBS will have a parameter c in the range $(0,1)$.

In the rest of this section, we prove our main result, Theorem~\ref{thm:main-GBS}, which shows that the
distribution in \eqref{eqn:GBS-G} can be (classically) sampled in polynomial time for any graph $G$ and $c>0$.

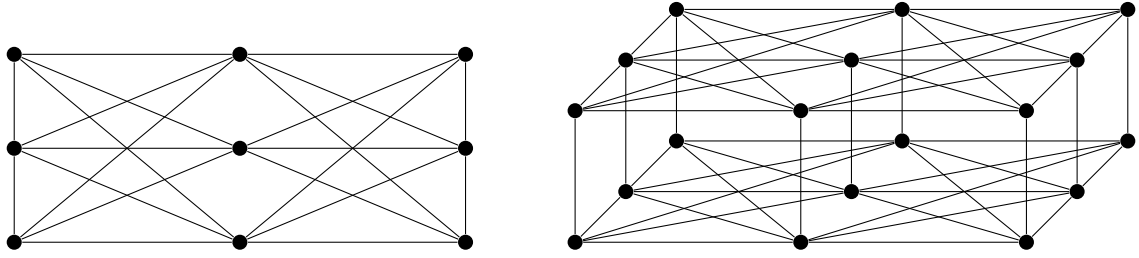
\begin{figure}[htbp]
    \centering
    \begin{tikzpicture}[every node/.style={circle, fill=black, minimum size=1.5mm, inner sep=0mm}]
    % Column 1 (Set A)
    \node (A1) at (0,3.75) {};
    \node (A2) at (0,2.5) {};
    \node (A3) at (0,1.25) {};

    % Column 2 (Set B)
    \node (B1) at (3,3.75) {};
    \node (B2) at (3,2.5) {};
    \node (B3) at (3,1.25) {};

    % Column 3 (Set C)
    \node (C1) at (6,3.75) {};
    \node (C2) at (6,2.5) {};
    \node (C3) at (6,1.25) {};

    % Connect adjacent S vertices (A to B, B to C, C to D)
    \foreach \i in {1,2,3} {
        \foreach \j in {1,2,3} {
            \draw (A\i) -- (B\j);
            \draw (B\i) -- (C\j);
        }
    }
    
    % Connect outer edges
    \draw (A1) -- (A2);
    \draw (A2) -- (A3);
    \draw (C1) -- (C2);
    \draw (C2) -- (C3);
\end{tikzpicture}
\hspace{1cm}
    % Second graph (3D) - Isometric Projection of two layers
\begin{tikzpicture}[xshift=12cm, scale=1, every node/.style={circle, fill=black, minimum size=1.5mm, inner sep=0mm}]
    % Layer 1 (Set A, B, C, D at z=0)
    % Column 1 (Set A)
    \node (A1) at (0,0,5.25) {};
    \node (A2) at (0,0,3.5) {};
    \node (A3) at (0,0,1.75) {};
    % Column 2 (Set B)
    \node (B1) at (3,0,5.25) {};
    \node (B2) at (3,0,3.5) {};
    \node (B3) at (3,0,1.75) {};

    % Column 3 (Set C)
    \node (C1) at (6,0,5.25) {};
    \node (C2) at (6,0,3.5) {};
    \node (C3) at (6,0,1.75) {};

    % Layer 2 (Set A', B', C', D' at z=2)
    \node (A1') at (0,1.75,5.25) {};
    \node (A2') at (0,1.75,3.5) {};
    \node (A3') at (0,1.75,1.75) {};

    \node (B1') at (3,1.75,5.25) {};
    \node (B2') at (3,1.75,3.5) {};
    \node (B3') at (3,1.75,1.75) {};

    \node (C1') at (6,1.75,5.25) {};
    \node (C2') at (6,1.75,3.5) {};
    \node (C3') at (6,1.75,1.75) {};

    % Connect adjacent columns in first layer (A to B, B to C, C to D)
    \foreach \i in {1,2,3} {
        \foreach \j in {1,2,3} {
            \draw (A\i) -- (B\j);
            \draw (B\i) -- (C\j);
        }
    }
    
    % Connect first layer outer edges
    \draw (A1) -- (A2);
    \draw (A2) -- (A3);
    \draw (C1) -- (C2);
    \draw (C2) -- (C3);

    % Connect adjacent columns in second layer (A to B, B to C, C to D)
    \foreach \i in {1,2,3} {
        \foreach \j in {1,2,3} {
            \draw (A\i') -- (B\j');
            \draw (B\i') -- (C\j');
        }
    }
    
    % Connect second layer outer edges
    \draw (A1') -- (A2');
    \draw (A2') -- (A3');
    \draw (C1') -- (C2');
    \draw (C2') -- (C3');

    % Now connect the two layers
    \foreach \i in {1,2,3} {
        \draw (A\i) -- (A\i');
        \draw (B\i) -- (B\i');
        \draw (C\i) -- (C\i');
    }

\end{tikzpicture}
\caption{On the left we have our original graph $G$, while on the right we have $G \boxprod K_2.$}
\label{fig:C-prod}
\end{figure}

The sampling algorithm that we use to prove Theorem~\ref{thm:main-GBS} 
takes the input graph~$G=(V,E)$ and constructs the graph $G\boxprod K_2$, which is the Cartesian product of $G$ and an edge. In order to construct $G\boxprod K_2$,
we start with a copy $G'=(V',E')$ of~$G$. 
Then $E_0 = \{  (v,v') : v\in V\}$
and $G\boxprod K_2 = (V\cup V', E \cup E' \cup E_0)$.
See Figure \ref{fig:C-prod} for an example.
Moreover, for each edge $e\in E\cup E'$, let $\lambda_e=c^2$, and for each $e\in E_0$, let $\lambda_e=1$.

Suppose that $M$ is a sample from the perfect matching distribution \eqref{eqn:PM-dist} on the graph $G\boxprod K_2$ with this weight function $\lambda$.
Let $S_M\subset V$ be the set of vertices that are endpoints of some edge in $M\cap E$. We have the following lemma, which shows that the
induced distribution of $S_M$ is given by Equation~\eqref{eqn:GBS-G}.

\begin{figure}[htbp]
  % Second graph (3D) - Isometric Projection of two layers
  \begin{tikzpicture}[xshift=12cm, scale=1, every node/.style={circle, fill=black, minimum size=1.5mm, inner sep=0mm}]
    % Layer 1 (Set A, B, C, D at z=0)
    % Column 1 (Set A)
    \node (A1) at (0,0,5.25) {};
    \node[fill=blue, minimum size=1.5mm] (A2) at (0,0,3.5) {};
    \node[fill=blue, minimum size=1.5mm] (A3) at (0,0,1.75) {};

    % Column 2 (Set B)
    \node[fill=blue, minimum size=1.5mm] (B1) at (3,0,5.25) {};
    \node[fill=blue, minimum size=1.5mm] (B2) at (3,0,3.5) {};
    \node[fill=blue, minimum size=1.5mm] (B3) at (3,0,1.75) {};

    % Column 3 (Set C)
    \node (C1) at (6,0,5.25) {};
    \node[fill=blue, minimum size=1.5mm] (C2) at (6,0,3.5) {};
    \node (C3) at (6,0,1.75) {};

    % Layer 2 (Set A', B', C', D' at z=2)
    \node (A1') at (0,1.75,5.25) {};
    \node[fill=red, minimum size=1.5mm] (A2') at (0,1.75,3.5) {};
    \node[fill=red, minimum size=1.5mm] (A3') at (0,1.75,1.75) {};

    \node[fill=red, minimum size=1.5mm] (B1') at (3,1.75,5.25) {};
    \node[fill=red, minimum size=1.5mm] (B2') at (3,1.75,3.5) {};
    \node[fill=red, minimum size=1.5mm] (B3') at (3,1.75,1.75) {};

    \node (C1') at (6,1.75,5.25) {};
    \node[fill=red, minimum size=1.5mm] (C2') at (6,1.75,3.5) {};
    \node (C3') at (6,1.75,1.75) {};

    % Connect adjacent columns in first layer (A to B, B to C, C to D)
    \foreach \i in {1,2,3} {
        \foreach \j in {1,2,3} {
            \draw[draw opacity=0.3, dashed] (A\i) -- (B\j);
            \draw[draw opacity=0.3, dashed] (B\i) -- (C\j);
        }
    }
    
    % Connect first layer outer edges
    \draw[draw opacity=0.3, dashed] (A1) -- (A2);
    \draw[draw opacity=0.3, dashed] (A2) -- (A3);
    \draw[draw opacity=0.3, dashed] (C1) -- (C2);
    \draw[draw opacity=0.3, dashed] (C2) -- (C3);

    % Connect adjacent columns in second layer (A to B, B to C, C to D)
    \foreach \i in {1,2,3} {
        \foreach \j in {1,2,3} {
            \draw[draw opacity=0.3, dashed] (A\i') -- (B\j');
            \draw[draw opacity=0.3, dashed] (B\i') -- (C\j');
        }
    }
    
    % Connect second layer outer edges
    \draw[draw opacity=0.3, dashed] (A1') -- (A2');
    \draw[draw opacity=0.3, dashed] (A2') -- (A3');
    \draw[draw opacity=0.3, dashed] (C1') -- (C2');
    \draw[draw opacity=0.3, dashed] (C2') -- (C3');

    % Connect adjacent S vertices (A to B, B to C, C to D)
    \draw[blue, thick] (A3) -- (B1);
    \draw[blue, thick] (A2) -- (B2);
    \draw[blue, thick] (B3) -- (C2);
    \draw[red, thick] (A3') -- (B2');
    \draw[red, thick] (A2') -- (B1');
    \draw[red, thick] (C2') -- (B3');

    % Now connect the two layers
    \foreach \i in {1,2,3} {
        \draw[draw opacity=0.3, dashed] (A\i) -- (A\i');
        \draw[draw opacity=0.3, dashed] (B\i) -- (B\i');
        \draw[draw opacity=0.3, dashed] (C\i) -- (C\i');
    }
    
    \draw[black, thick] (A1) -- (A1');
    \draw[black, thick] (C3) -- (C3');
    \draw[black, thick] (C1) -- (C1');
  \end{tikzpicture}
  \hspace{1cm}
    \centering
  \begin{tikzpicture}[every node/.style={circle, fill=black, minimum size=1.5mm, inner sep=0mm}]
    % Column 1 (Set A)
    \node[fill=purple!70, minimum size=2mm] (A1) at (0,3.75) {};
    \node[fill=purple!70, minimum size=2mm] (A2) at (0,2.5) {};
    \node (A3) at (0,1.25) {};

    % Column 2 (Set B)
    \node[fill=purple!70, minimum size=2mm] (B1) at (3,3.75) {};
    \node[fill=purple!70, minimum size=2mm] (B2) at (3,2.5) {};
    \node[fill=purple!70, minimum size=2mm] (B3) at (3,1.25) {};

    % Column 3 (Set C)
    \node (C1) at (6,3.75) {};
    \node[fill=purple!70, minimum size=2mm] (C2) at (6,2.5) {};
    \node (C3) at (6,1.25) {};

    % Connect adjacent S vertices (A to B, B to C, C to D)
    \foreach \i in {1,2,3} {
        \foreach \j in {1,2,3} {
            \draw[draw opacity=0.3, dashed] (A\i) -- (B\j);
            \draw[draw opacity=0.3, dashed] (B\i) -- (C\j);
        }
    }
    
    % Connect outer edges
    \draw[draw opacity=0.3, dashed] (A1) -- (A2);
    \draw[draw opacity=0.3, dashed] (A2) -- (A3);
    \draw[draw opacity=0.3, dashed] (C1) -- (C2);
    \draw[draw opacity=0.3, dashed] (C2) -- (C3);
    
    % Connect outer S vertices (A to B, B to C, C to D)%
%    \draw[draw opacity=0.3, dashed] (A1) -- (A2);
    \draw[red, thick] (A1) -- (B2);
    \draw[blue, thick] (A1) -- (B3);
    \draw[red, thick] (A2) -- (B3);
    \draw[blue, thick] (A2) -- (B2);
    \draw[purple!70, thick] (B1) -- (C2);
  \end{tikzpicture}
\caption{On the left, we have a perfect matching $M$ in $G \boxprod K_2$, with $M \cap E$ highlighted in red in the top copy, $M\cap E'$ blue in the bottom copy, and $M \cap E_0$ black.
On the right, we have our underlying graph $G$ with the set $S_M$ and the corresponding edges chosen by $M$ (in either or both copies) highlighted.}
\label{fig:K-2-matchings}
\end{figure}
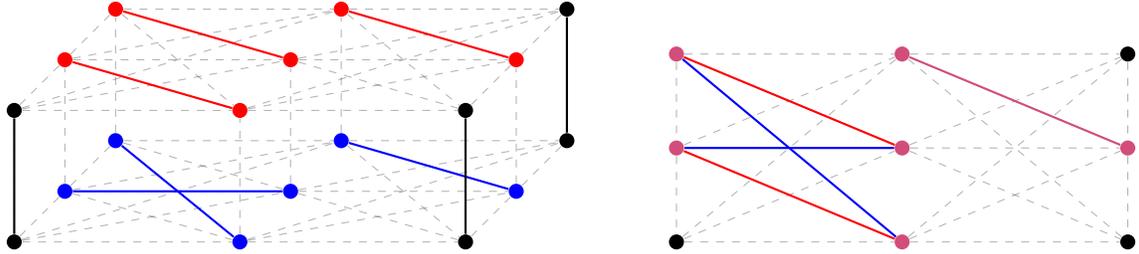

\begin{lemma}\label{lem:ind-dist}
    The induced distribution of $S_M$ is exactly 
  $\mu_{GBS,G}(S) $.
\end{lemma}
\begin{proof}
For $M\sim\mu_{PM,\lambda}$, as it is a perfect matching of $G\boxprod K_2$,
each vertex $v$ in $V$ is either matched with an edge in $M\cap E$ or $M\cap E_0$.
If $v$ is matched in~$E_0$ then it is matched with its copy $v'$. Thus the set of endpoints of edges in $M\cap E'$, denoted $S'$, is the set of copies of vertices in~$S_M$.
See Figure \ref{fig:K-2-matchings} for an illustration.
    
To summarise, the probability  of outputting a particular $S\subset V$ is
    \begin{align*}
        \Pr[M\sim\mu_{PM,\lambda}]{S_M=S} & \propto c^{2\abs{S}} \Haf[A_S] \Haf[A_{S'}]\\
        & = c^{2\abs{S}} \Haf[A_S]^2 \qedhere
    \end{align*}
\end{proof}

Our problem now is to sample from $\mu_{PM,\lambda}$ for $G \boxprod K_2$. 
For this, we use Proposition~\ref{prop:JS89}, 
which can be used to efficiently sample perfect matchings in only if the ratio between the number of perfect and near-perfect matchings is bounded from above by a polynomial.
We will  show that this is indeed the case for $G\boxprod K_2$.

To be more specific, 
define a new weight function $\lambda'$ for $G \boxprod K_2$ as $\lambda_e'=4n^2\lambda_e$.
This weight function favours matchings of larger sizes, and will eventually be the weight function that we will use when invoking the Jerrum--Sinclair chain \cite{JS89}.
For a matching $M$ of $G\boxprod K_2$, let $w(M)\defeq\prod_{e\in M}\lambda_e$ and $w'(M)\defeq\prod_{e\in M}\lambda'_e$.
Let $M_k$ denote the set of matchings  with $k$ edges in $G\boxprod K_2$, where $0\le k\le n$. 
Let $Z_k$ be the total weight from $k$-edge matchings under $\lambda$, namely,
\begin{align*}
    Z_k\defeq\sum_{M\in M_k}w(M),
\end{align*}
and $Z\defeq\sum_{k=1}^n Z_k$ be the partition function of $\mu_{matching,\lambda}$.
Similarly, define $Z_k'$ and $Z'$ for $\lambda'$.
Note that the proof of Lemma~\ref{lem:ind-dist} implies that
\begin{align}\label{eqn:sum-Z_n}
    Z_{n} = \sum_{S\subseteq V}c^{2\abs{S}}\Haf(A_S)^2.
\end{align}

\begin{lemma}\label{lem:PM-vs-NPM} 
    $Z_{n-1}< 2n^2 Z_{n}$.
\end{lemma}
\begin{proof}
We write $\Haf(S)$ for $\Haf(A_S)$.
Note that by choosing the two vertices to remove, we have
\begin{align}
Z_{n-1} &= 2 \sum_{(u_1,u_2) \in \binom{V}{2}} \sum_{S \subseteq V \setminus \{u_1,u_2\}} c^{2\abs{S}+2}\Haf(S) \cdot \Haf(S \cup \{u_1,u_2\})\notag\\
    &\quad +\sum_{u_1,u_2 \in V} \sum_{S \subseteq V \setminus \{u_1,u_2\}} c^{2\abs{S}+2}\Haf(S \cup \{u_1\}) \cdot \Haf(S \cup \{u_2\}).\label{eqn:Z_n-1-sum}
    \end{align}
By the AM-GM inequality, it holds that
\begin{align*}
        &\quad \sum_{(u_1,u_2) \in \binom{V}{2}} \sum_{S \subseteq V \setminus \{u_1,u_2\}} c^{2\abs{S}+2}\Haf(S) \cdot \Haf({S \cup \{u_1,u_2\}}) \\
        &\le \frac{1}{2}\sum_{(u_1,u_2) \in \binom{V}{2}} \sum_{S \subseteq V \setminus \{u_1,u_2\}} \left(c^{2\abs{S}}\Haf(S)^2+ c^{2\abs{S}+4}\Haf(S \cup \{u_1,u_2\})^2\right)\\
        &\le \binom{n}{2} \sum_{S \subseteq V} c^{2\abs{S}}\Haf(S)^2 = \binom{n}{2}Z_{n},
    \end{align*}
    where we used \eqref{eqn:sum-Z_n} in the last line.
    Similarly,
    \begin{align*}
        &\quad \sum_{u_1,u_2 \in V} \sum_{S \subseteq V \setminus \{u_1,u_2\}} c^{2\abs{S}+2}\Haf(S \cup \{u_1\}) \cdot \Haf(S \cup \{u_2\}) \\
        & \le \frac{1}{2} \sum_{u_1,u_2 \in V} \sum_{S \subseteq V \setminus \{u_1,u_2\}} \left(c^{2\abs{S}+2}\Haf(S \cup \{u_1\})^2+ c^{2\abs{S}+2}\Haf(S \cup \{u_2\})^2\right)\\ 
        &\le n^2 \cdot \sum_{S \subseteq V} c^{2\abs{S}}\Haf(S)^2=n^2 Z_{n}.
    \end{align*}
    Combining with \eqref{eqn:Z_n-1-sum}, it follows that
    \begin{align*}
    	Z_{n-1} &< 2 n^2  Z_{n}. \qedhere 
    \end{align*}
\end{proof}

In addition, Heilmann and Lieb \cite[Theorem 7.1]{HL72} showed that the sequence $(Z_k)$ is log-concave.
This is a weighted version of Proposition \ref{prop:log-concave}.

\begin{proposition}\label{prop:Z-log-concave}
  For any $1\le k\le n-1$, $Z_{k-1}Z_{k+1}\le Z_{k}^2$.
\end{proposition}

By Proposition~\ref{prop:Z-log-concave} and Lemma~\ref{lem:PM-vs-NPM}, we have that $Z_{k-1}\le 2n^2Z_k$ for any $1\le k\le n$.
Since $Z_k'=(4n^2)^k Z_k$, it follows that for any $1\le k\le n$, $Z_{k-1}'\le Z_k'/2$.
Thus,
\begin{align}\label{eqn:PM-percent}
    Z'=\sum_{i=0}^nZ_i'\le \sum_{i=0}^n 2^{i-n}Z_n' <2Z_n'.
\end{align}

We can now prove Theorem~\ref{thm:main-GBS}.
\mainGBS*

\begin{proof} 
By Lemma~\ref{lem:ind-dist}, we just need to sample from $\mu_{PM,\lambda}$ over perfect matchings in $G\boxprod K_2$.
  To this end, we use Proposition~\ref{prop:JS89} to approximately sample from the distribution $\mu_{matching,\lambda'}$ in \eqref{eqn:matching} for $G\boxprod K_2$ with the weight function $\lambda'$ as described above.
  We set the error as $\eps'=\min\{1/4,\eps/2\}$ in Proposition~\ref{prop:JS89}.
  By \eqref{eqn:PM-percent}, there is $1/2-\eps'\geq1/4$ probability that the algorithm outputs a perfect matching. 
  We keep rejecting until we see a perfect matching.
  The distribution $\mu_{matching,\lambda'}$ conditioned on outputting a perfect matching is exactly $\mu_{PM,\lambda}$.
  Thus we can approximately sample from $\mu_{PM,\lambda}$ with $\eps$ error after at most $O(\log{(1/\eps)})$ rejections.
  As for the running time, we plug in $\lambda = O(n^2)$ in Proposition~\ref{prop:JS89},
  which finishes the proof.  
\end{proof}

\section{Concluding remarks}
In this paper we showed that there are polynomial-time classical algorithms to (approximately) sample from the distribution \eqref{eqn:GBS-G}, related to Gaussian Boson Sampling on graphs, and from the distribution \eqref{eqn:BS} when the input matrix $A$ is non-negative.
It would be interesting to know whether the non-negative restriction can be relaxed.
While the algorithm 
on which
our method is based works only for non-negative matrices,
in the more general case, our construction actually enables a reduction from sampling to approximating permanents of some complex weighted matrices.
This is because, as can be seen from the proof of Theorem \ref{thm:main-BS}, the marginal probability of choosing a row is the ratio between the permanents of two complex weighted matrices, even when conditioned on previous choices.
Thus, one can sample from the (conditional) marginal distributions of the rows, one by one, using an oracle that approximates the permanent of complex weighted matrices.
However, the latter task is \textbf{NP}-hard in general \cite{Mei23},
and it is an interesting open question to what extent this strategy can be useful.

\section*{Author contributions}
All authors contributed equally to this work.
No AI used.

\section*{Acknowledgement}

HG would like to thank Raul Garcia-Patron for pointing out a few useful references.

This project has received funding from the European Research Council (ERC) under the European Union's Horizon 2020 research and innovation programme (grant agreement No.~947778).

\bibliographystyle{quantum}
\bibliography{ref}
\end{document}